\documentclass[11pt]{article}
\usepackage[left=1in, right=1in, top=1in, bottom=1in, margin=1in]{geometry}

\usepackage{tcolorbox}
\usepackage{ninecolors}
\usepackage{xcolor}

\usepackage{amsmath,amssymb,xspace,graphicx,relsize,bm,breqn,multirow}
\usepackage{multicol}

\usepackage{braket}
\usepackage{qcircuit}
\usepackage{adjustbox}

\usepackage{float}
\usepackage[linesnumbered,ruled,vlined]{algorithm2e}
\SetKwInput{KwInput}{Input}
\SetKwInput{KwOutput}{Output}
\SetKwInOut{Promise}{Promise}
\SetKwInput{Goal}{Goal}
\SetKwProg{Fn}{function}{}{}
\SetKwFor{RepTimes}{repeat}{times}{}
\SetKwFunction{Ver}{Verify}
\SetKwFunction{Prep}{PrepareState}
\SetKwComment{Comment}{/* }{ */}
\usepackage{algorithmic}

\newtcolorbox{myalgorithm}[1][]{
    colback=gray!10, 
    colframe=black, 
    arc=5pt, 
    boxrule=0.5pt, 
    left=0pt, right=0pt, top=0pt, bottom=0pt 
}
\newcommand{\Exp}{\mathop{\mathbb{E}}}

\usepackage[margin=1in]{geometry}

\usepackage{amsthm}
\usepackage{dsfont}
\usepackage{array}
\usepackage{makecell}

\newcommand{\poly}{\ensuremath{\mathsf{poly}}}

\newcommand{\id}{\ensuremath{\mathbb{I}}}

\usepackage{upgreek}
\usepackage{enumerate}

\usepackage[pagebackref]{hyperref}
\usepackage[usestackEOL]{stackengine}
\usepackage{thm-restate,mathrsfs}
\usepackage{enumerate}
\usepackage{array}
\usepackage{parskip}
\def\01{\{0,1\}}

\hypersetup{
	colorlinks,
	linkcolor={blue!100!black},
	citecolor={red!100!black},
}

\newcommand{\be}{\begin{equation}}
\newcommand{\ee}{\end{equation}}
\newcommand{\ba}{\begin{array}}
\newcommand{\ea}{\end{array}}
\newcommand{\bea}{\begin{eqnarray}}
\newcommand{\eea}{\end{eqnarray}}

\usepackage{mathtools}

\newcommand{\ra}{\rangle}
\newcommand{\la}{\langle}

\newcommand{\norm}[1]{\left\lVert#1\right\rVert}

\newcommand{\polylog}{\textsf{polylog}}
\newcommand{\Stab}{\textsf{Stab}}
\newcommand{\proj}[1]{\ket{#1}\!\bra{#1}}

\usepackage{amsthm}

\def\01{\{0,1\}}

\definecolor{citegreen}{HTML}{208054}
\definecolor{citeblue}{HTML}{0055cc}

\NineColors{saturation=high}
\hypersetup{
    breaklinks=true,   
    colorlinks=true, 
    linkcolor=blue3, 
    citecolor=green5, 
    urlcolor=blue3, 
}

\newtheorem{theorem}{Theorem}[section]

\newtheorem{lemma}[theorem]{Lemma}

\newtheorem{corollary}[theorem]{Corollary}

\newtheorem{fact}[theorem]{Fact}
\newtheorem{claim}[theorem]{Claim}

\usepackage{thm-restate,mathrsfs} 

\usepackage{tcolorbox}

\makeatletter
\def\widebreve{\mathpalette\wide@breve}
\def\wide@breve#1#2{\sbox\z@{$#1#2$}%
     \mathop{\vbox{\m@th\ialign{##\crcr
\kern0.08em\brevefill#1{0.8\wd\z@}\crcr\noalign{\nointerlineskip}%
                    $\hss#1#2\hss$\crcr}}}\nolimits}
\def\brevefill#1#2{$\m@th\sbox\tw@{$#1($}%
    \hss\resizebox{#2}{\wd\tw@}{\rotatebox[origin=c]{90}{\upshape(}}\hss$}
\makeatletter

\title{A note on tomography of states with low stabilizer rank}

\begin{document}

\author{
Arkopal Dutt\\
\small IBM Research, Cambridge \\
\small \texttt{arkopal@ibm.com}
}
\date{}
\maketitle
\begin{abstract}
In this note, we consider tomography of $n$-qubit quantum states $|\psi \rangle$ that are promised to have stabilizer rank $k$ i.e., admit a stabilizer decomposition of $\ket{\psi} = \sum_{i=1}^k c_i \ket{\phi_i}$ where each $\ket{\phi_i}$ is a stabilizer state. We give a tomography protocol for such states up to trace distance $\varepsilon > 0$ running in $\textsf{poly}(n,2^k,1/\varepsilon)$ time and can thus handle $k=O(\log n)$ efficiently. This follows from first showing that the stabilizer extent $\xi(\ket{\psi})$, which is the minimal $\ell_1$ norm over the coefficients $\{c_i\}$ of any stabilizer decomposition of $\ket{\psi}$, can be bounded as $\xi(\ket{\psi}) \leq 2^{O(k)}$ and improving upon the previous best bound of $k^{O(k)}$. The protocol is then obtained by utilizing a tomography protocol for states with bounded stabilizer extent.
\end{abstract}

\section{Introduction}\label{sec:intro}
Great strides have been made in obtaining efficient quantum state tomography protocols for quantum states that are promised to have compact stabilizer structure. In this note, we consider the setting when the unknown $n$-qubit quantum state is promised to have stabilizer rank $k$ or can be expressed as a superposition of at most $k$ many stabilizer states:
$$
\ket{\psi} = \sum_{i=1}^k c_i \ket{\phi_i}, \quad \text{where} \ket{\phi_i} \in \Stab, \,\, \text{ for all } i \in [k],
$$
where $\Stab$ denotes the set of $n$-qubit stabilizer states. The coefficients $c_i \in \mathbb{C}$ and the stabilizer components $\{\ket{\phi_i}\}$ need not be orthogonal. The notion of stabilizer rank was originally introduced for understanding classical simulability limits of quantum states and circuits~\cite{bravyi2016trading,bravyi2019simulation}. We have polynomial-time algorithms available when $k=1$ or $\ket{\psi}$ is itself a stabilizer state~\cite{montanaro2017learning,arunachalam2022optimal,arunachalam2026optimal,chang2026sample} and all the stabilizer components $\{\ket{\phi_i}\}$ are promised to have a large subgroup of Paulis stabilizing them (also called states with high stabilizer dimension)~\cite{grewal2023efficient,leone2024learningstab,hangleiter2024bell,chia2024efficient}. Recently, it was shown that states with stabilizer extent $\xi$, originally introduced in~\cite{bravyi2019simulation} and defined as the minimum~$\ell_1$ norm of the coefficients over all possible stabilizer decompositions of $\ket{\psi}$, could be learned in $\poly(n,\xi,1/\varepsilon)$~\cite{ad2026structure,ad2026tomoextent,adgg2026apfr}. However, this translates to a $\poly(n,k^k, 1/\varepsilon)$-time tomography protocol of stabilizer rank-$k$ states using known upper bounds on stabilizer extent for rank-$k$ states~\cite{kalra2025stabilizer}. The natural question is: Can we carry out tomography of stabilizer rank-$k$ states in $\poly(n,k,1/\varepsilon)$-time? In this note, I make progress on this question and show the~following.
\begin{theorem}\label{thm:tomo_rank}
Let $\varepsilon,\delta \in (0,1)$ and $k \in \mathbb{N}$. Suppose $\ket{\psi}$ is an unknown $n$-qubit state that is promised to have stabilizer rank $k$. Then, there is an algorithm that with probability $\geq 1-\delta$ outputs $\ket{\phi}$ which has stabilizer rank $\poly(2^k/\varepsilon)$ and satisfies 
$$
|\la \psi | \phi \ra|^2 \geq 1 - \varepsilon.
$$
The algorithm uses $\poly(n,2^k,1/\varepsilon,\log(1/\delta))$ time and copies.
\end{theorem}
\noindent To prove this, we give an improved upper bound on the stabilizer extent of stabilizer rank-$k$ states. This improves upon the existential result of~\cite{mehraban2024improved} and the previous best bound of $k^{O(k)}$~\cite{kalra2025stabilizer}.
\begin{theorem}\label{lem:ub_extent_rank}
Let $k \in \mathbb{N}$. For a stabilizer-rank $k$ state $\ket{\psi}$, its stabilizer extent is bounded as~$\xi(\ket{\psi}) \leq~2^{O(k)}$.    
\end{theorem}

\section{Proof of main result}\label{sec:proof}
\subsection{Preliminaries}

\paragraph{Notation.} 
We write $[m]=\{1,\ldots,m\}$,  use $\|\cdot\|$ for the Euclidean norm of a vector and $\|\cdot\|_1$ for the trace norm of an operator. We will work with the $n$-qubit Pauli group of $\{\pm i, \pm 1\} \times \{\id,X,Y,Z\}^{\otimes n}$ where the single qubit Pauli operators of $\id,X,Y,Z$ have their usual meanings. For a stabilizer state $|\phi\rangle$, let $G_\phi$ be its signed stabilizer group, consisting of the $2^n$ commuting Hermitian Paulis satisfying $P|\phi\rangle=|\phi\rangle$. For independent generators $P_1,\ldots,P_n$ of $G_\phi$, we can express the Pauli decomposition of the density matrix corresponding to~$\ket{\phi}$~as
\begin{equation}\label{eq:stab_density_matrix}
 |\phi\rangle\langle\phi| = \prod_{j=1}^n\frac{I+P_j}{2} = \frac1{2^n}\sum_{P\in G_\phi}P.
\end{equation}
For a Hermitian Pauli $P$, we also write $\Pi_P^\pm=(I\pm P)/2$ for its two measurement-outcome projectors. We will often use the following fact.
\begin{fact}\label{fact:nv2-postselection}
For any Hermitian Pauli $P$ and $n$-qubit stabilizer state $|\phi\rangle$, each vector $\Pi_P^\pm|\phi\rangle$ is zero or a scalar multiple of a stabilizer state.
\end{fact}
\begin{proof}
If $P$ or $-P$ belongs to $G_\phi$, then $|\phi\rangle$ is an eigenstate
of $P$, and each projected vector is either zero or $|\phi\rangle$.
Otherwise $P$ anticommutes with some stabilizer generator. Choose
generators $g_1,\ldots,g_n$ so that $g_1$ anticommutes with $P$ and all
the others commute with $P$. Such a choice is obtained by multiplying
any other anticommuting generator by $g_1$.

Anti-commutation implies $\langle\phi|P|\phi\rangle=0$, so both outcomes have probability $1/2$. The outcome of $\Pi_P^\pm|\phi\rangle$ is stabilized by $\pm P,g_2,\ldots,g_n$, a set of $n$ independent commuting Pauli
generators. It is therefore a scalar multiple of a stabilizer state. 
\end{proof}

\paragraph{Stabilizer rank and extent.}
For any $n$-qubit state vector $\psi\in(\mathbb C^2)^{\otimes n}$, not necessarily normalized, we define its stabilizer rank denoted by $\chi(\psi)$ and stabilizer extent denoted by $\chi(\psi)$ as
\begin{align}
\chi(\psi) &:=\min\left\{r:\psi=\sum_{i=1}^r c_i|\phi_i\rangle,\
                  |\phi_i\rangle\in\Stab\right\},
       \label{eq:nv2-rank}\\
\xi(\psi)
&:=\min\Big\{\sum_i|c_i|: \psi=\sum_i c_i|\phi_i\rangle,\
                |\phi_i\rangle\in\Stab\Big\}.
       \label{eq:nv2-extent}
\end{align}
In this note, $\xi$ is the \emph{unsquared} coefficient $\ell_1$ norm. We will use the convention of $\xi(0)=\chi(0)=0$. For $z\ne0$, $\chi(z\psi)=\chi(\psi)$. We also have the following properties of stabilizer extent which can be checked by observation:
\begin{equation}\label{eq:nv2-norm-properties}
 \xi(z\psi)=|z|\xi(\psi),\qquad
 \xi(\psi+\omega)\le\xi(\psi)+\xi(\omega),\qquad
 \xi(|\phi\rangle)=1\quad(|\phi\rangle\in\operatorname{Stab}).
\end{equation}
Indeed, every decomposition satisfies $\|\psi\|\le\sum_i|c_i|$ by the triangle inequality, so $\|\psi\|\le\xi(\psi)$; the one-term expansion proves the last equality in \eqref{eq:nv2-norm-properties}. We will use the following prior tomography result that gives a tomography protocol for states promised to have bounded stabilizer extent.
\begin{theorem}[\cite{adgg2026apfr,ad2026tomoextent}]\label{thm:tomo_extent}
Let $\varepsilon,\delta \in (0,1)$ and $\xi > 0$. Suppose $\ket{\psi}$ is an unknown $n$-qubit state with $\xi(\ket{\psi}) \leq \xi$. Then, there is an algorithm that with probability $\geq 1-\delta$, outputs a $\poly(\xi/\varepsilon)$-stabilizer-rank state $\ket{\phi}$ such that
$$
|\la \psi | \phi \ra|^2 \geq 1 - \varepsilon.
$$
The algorithm has $\poly(n,\xi,1/\varepsilon,\log(1/\delta))$ sample and time complexity.
\end{theorem}

\subsection{Upper bound on stabilizer extent}
We now prove Theorem~\ref{lem:ub_extent_rank} by showing a slightly more general result from which it will directly follow.
\begin{lemma}\label{lem:ub_extent}
Let $k \in \mathbb{N}$. For every $n$-qubit vector $\psi$ with $\chi(\psi) \le k$,
$$
\xi(\psi) \leq 2^{O(k)} \norm{\psi}.
$$
Every normalized quantum state $\ket{\psi}$ with stabilizer rank $k$ then satisfies $\xi(\ket{\psi}) \leq 2^{O(k)}$.
\end{lemma}
\begin{proof}
We will prove this via induction on $k$. For $\psi$ being the zero vector, the result is immediate. For $k=1$, $\psi$ is a scalar multiple of a stabilizer state, and hence $\xi(\psi)=\|\psi\|$. Suppose $k\geq 2$. Our induction hypothesis is that for every $n$-qubit vector $\psi$
\begin{equation}\label{eq:induction-hypothesis}
\chi(\psi)\le k-1 \quad\Longrightarrow\quad \xi(\psi) \leq C_{k-1}\,\|\psi\|,
\end{equation}
where $C_{k-1} \geq 1$ is some constant, which we will fix later, that depends only on $(k-1)$. We now prove the corresponding implication with $C_k$ for $\chi(\psi)\leq k$. Let $\psi$ have the following stabilizer decomposition
\begin{equation}\label{eq:stab_decomp_psi}
\psi = \sum_{i=1}^k c_i \ket{\phi_i}, \quad \ket{\phi_i} \in \Stab \,\,\forall i \in [k]
\end{equation}
Starting from the above decomposition, our goal will be to now obtain an alternate stabilizer decomposition of $\psi$ which may have a much higher stabilizer rank than $k$ but for which we can bound the stabilizer extent. 

We will do this iteratively over $\kappa \leq n$ many iterations. We will define the stopping criterion shortly. We first fix $\ket{\phi_1}$ as a reference stabilizer component with respect to which we will define each iteration. Let us denote $\beta_1 = \la \phi_1 | \psi \ra$ and let $G$ be the Abelian group of Paulis that stabilize $\ket{\phi_1}$. In each iteration $t$, starting from iteration $t = 1$ and up to $n$, we will sequentially define Paulis $P_t \in G$ that are independent of $\{P_1,\ldots,P_{t-1}\}$. Moreover, we define the \emph{residual} state vector at the end of iteration $t$ after determining $P_t$ as
\begin{equation}
\psi_{t+1} := \Pi_{P_t}^+ \psi_t
\end{equation}
where $\Pi_P^{\pm} = (\id \pm P)/2$ and we will define the initial residual state $\psi_1$ as
$$
\psi_1 := (\id - \ket{\phi_1}\!\bra{\phi_1})\psi = \psi - \beta_1 \ket{\phi_1}.
$$
Note that by definition, $\la \psi_1 | \phi_1 \ra = 0$. Corresponding to each residual state vector $\psi_t$, we define $\sigma_{t+1} := (\id - \Pi_{P_t}^+) \psi_t = \Pi_{P_t}^{-} \psi_t$ with $\sigma_1 = \beta_1 \ket{\phi_1}$. Note that $\la \sigma_t | \psi_t \ra = 0, \forall t \in [n]$ by definition. We stop in iteration $\kappa$ when the residual state vector obtained at the end satisfies
\begin{equation}\label{eq:stopping_criterion}
\norm{\psi_{\kappa + 1}}^2 = 0
\end{equation}
We now discuss how $P_t$ is determined in an iteration $t$ and the resulting properties of $\psi_{t+1}$ and~$\sigma_{t+1}$.

\begin{claim}\label{claim:norm_drop}
Suppose we are in iteration $t \geq 1$ and have not yet stopped. Then, there exists a Pauli $P_t \neq \id \in G$ independent of $\{P_i\}_{i \in [t-1]}$ such that 
\begin{enumerate}[$(i)$]
    \item $\norm{\psi_{t+1}}^2 \leq \norm{\psi_t}^2/2$, and
    \item $\chi(\sigma_{t+1}) \leq k - 1$.
\end{enumerate}
\end{claim}
\begin{proof}
For a uniformly random $P \in G$, we observe that
\begin{align}
\Exp_{P \in G} \left[ \norm{\Pi_P^+ \psi_t}^2 \right] &= \Exp_{P \in G} \left[ \la \psi_t | \Pi_{P}^+ | \psi_t \right] \nonumber \\
&= \left\la \psi_t \left| \frac{\id}{2} + \frac{1}{2^{n+1}} \sum_{P \in G} P \right| \psi_t \right\ra \nonumber \\
& =\left\langle \psi_t\left| \frac{I+\proj{\phi_1}}2\right|\psi_t\right\rangle \nonumber \\
& = \frac{1}{2} \norm{\psi_t}^2, \label{eq:exp_norm_residual}
\end{align}
where we used $\Pi_P^+$ is an orthogonal projector in the first equality, Eq.~\eqref{eq:stab_density_matrix} in the third equality and finally the fact that $\la \phi_1 | \psi_t \ra = 0$ which is obtained from direct evaluation:
$$
\la \phi_1 | \psi_t \ra = \la \phi_1 | \Big( \Pi_{s=1}^{t-1} \Pi_{P_s}^+ \Big) \ket{\psi_1} = \la \phi_1 | \psi_1 \ra = 0,
$$
where we used the definition of $\psi_t$ in the first equality, that $P_s \ket{\phi_1} = \ket{\phi_1}$ since $P_s \in G, \,\,\forall s \in [t-1]$ in the second equality and that $\la \phi_1 | \psi_1 \ra = 0$ by definition of $\psi_1$. 

From Eq.~\eqref{eq:exp_norm_residual}, we now note that there must be at least one Pauli $P_t$ such that
$$
\norm{\Pi_{P_t}^+\psi_t}^2\leq \|\psi_t\|^2/2
$$
Since $\psi_t \neq 0$ (or else we would have already stopped), this Pauli $P_t$ cannot be $\id$. Moreover, every element of the subgroup $G_t := \la \{P_i\}_{i \in [t-1]} \ra$ generated by $\{P_i\}_{i \in [t-1]}$ satisfies $\Pi_Q^{+} \psi_t = \psi_t, \, \forall Q \in G_t$ and therefore would leave its squared norm unchanged. Since $\psi_t\ne0$, $P_t$ must be independent of the previous choices and is not the identity. This proves item $(i)$.

For item $(ii)$, we now evaluate the expression of $\sigma_{t+1} = \Pi_{P_t}^{-} \psi_t$ which is
\begin{align}
\sigma_{t+1} =\Pi_{P_t}^- \left(\prod_{h=1}^{t-1}\Pi_{P_h}^+\right)\psi =\sum_{i=2}^k c_i\,
  \Pi_{P_t}^-\left(\prod_{h=1}^{t-1}\Pi_{P_h}^+\right) \ket{\phi_i},
\label{eq:v5-rankdrop}
\end{align}
where we noted that $P_h \ket{\phi_1} = \ket{\phi_1}, \forall h \in [t-1]$ and $\Pi_{P_t}^-\ket{\phi_1}=0$. Using Fact~\ref{fact:nv2-postselection}, each of the remaining $k-1$ projected components is zero or a scalar multiple of a stabilizer. Absorbing these scalars, including their phases, into the coefficients gives a stabilizer decomposition of $\sigma_{t+1}$ with at most $k-1$ terms. This proves $(ii)$ and the overall proof.
\end{proof}

We now complete the induction using the above claim. We first note that $\kappa \leq n$ since running $n$ iterations would generate $n$ independent Paulis generating $G$, and hence we would have
\[
\psi_{n+1}
=\left(\prod_{h=1}^n\Pi_{P_h}^+\right)\psi_1
=\proj{\phi_1}\psi_1=0.
\]
Thus the stopping index is well defined and $\kappa\le n$. Iterating the definitions of $\psi_t,\sigma_t$, together with the fact that initially $\psi=\sigma_1+\psi_1$, gives
\begin{equation}\label{eq:v5-reconstruction}
\psi=\sum_{j=1}^{t+1}\sigma_j+\psi_{t+1}
\quad(0\le t\le\kappa),
\qquad
\psi=\sum_{j=1}^{\kappa+1}\sigma_j.
\end{equation}
Here $\sigma_1$ is the contribution of the reference stabilizer state $\ket{\phi_1}$, while each
subsequent $\sigma_j$ has rank at most $k-1$ by Claim~\ref{claim:norm_drop}. Note that these subsequent vectors need not individually be stabilizer states but they are promised to have stabilizer rank $\leq k-1$ to which we can apply the induction. Orthogonality of $\psi_t,\sigma_t$ and item $(i)$ of Claim~\ref{claim:norm_drop} gives for each $1\le t\le\kappa$,
\begin{equation}\label{eq:v5-energy}
\|\sigma_{t+1}\|^2=\|\psi_t\|^2-\|\psi_{t+1}\|^2,
\qquad
\|\psi_{t+1}\|\le\frac{\|\psi_t\|}{\sqrt2}.
\end{equation}
Using Eq.~\eqref{eq:v5-energy}, we obtain
\begin{equation}\label{eq:v5-ratio}
\frac{\|\sigma_{t+1}\|}{\|\psi_t\|-\|\psi_{t+1}\|}
=\sqrt{\frac{\|\psi_t\|+\|\psi_{t+1}\|}
                 {\|\psi_t\|-\|\psi_{t+1}\|}}
\le\sqrt{\frac{1+1/\sqrt2}{1-1/\sqrt2}}=1+\sqrt2,
\end{equation}
where we have used that $\psi_t\ne0$ at every executed iteration so the denominator in the above ratio is positive. Consequently, we then have that
\begin{align}
\sum_{t=1}^{\kappa}\|\sigma_{t+1}\| \le(1+\sqrt2)\sum_{t=1}^{\kappa}
       \bigl(\|\psi_t\|-\|\psi_{t+1}\|\bigr) =(1+\sqrt2)\|\psi_1\|
 =(1+\sqrt2)\sqrt{\|\psi\|^2-|\beta_1|^2},
\label{eq:v5-norm-sum}
\end{align}
where we have used $\kappa \leq n$, used $\psi_{\kappa+1}=0$ in the final equality along with the initial orthogonal projection. Now, using item $(ii)$ of Claim~\ref{claim:norm_drop}, we have from Eq.~\eqref{eq:induction-hypothesis} that
\[
\xi(\sigma_{t+1})\le C_{k-1}\|\sigma_{t+1}\|.
\]
Using the decomposition of Eq.~\eqref{eq:v5-reconstruction}, the triangle inequality for $\xi$, and Eq.~\eqref{eq:v5-norm-sum}, we then obtain
\begin{align}
\xi(\psi) \le\xi(\sigma_1)+\sum_{t=1}^{\kappa}\xi(\sigma_{t+1}) \leq |\beta_1|+C_{k-1}\sum_{t=1}^{\kappa}\|\sigma_{t+1}\| 
&\le|\beta_1|+(1+\sqrt2)C_{k-1} \sqrt{\|\psi\|^2-|\beta_1|^2}\notag\\*
&\le\sqrt{1+(3+2\sqrt2)C_{k-1}^{\,2}}\,\|\psi\|,
\label{eq:v5-induction}
\end{align}
where in the last inequality, we used the Cauchy--Schwarz inequality. Note that it suffices to choose
\begin{equation}\label{eq:v5-recurrence}
C_1=1,\qquad
C_k=\sqrt{1+(3+2\sqrt2)C_{k-1}^{\,2}}\quad(k\ge2).
\end{equation}
Squaring this recurrence and iterating yields
\begin{equation}\label{eq:v5-constant}
C_k^2=\sum_{j=0}^{k-1}(3+2\sqrt2)^j,
\qquad
C_k=\sqrt{\frac{(3+2\sqrt2)^k-1}{2+2\sqrt2}}
\le(1+\sqrt2)^k=2^{O(k)}.
\end{equation}
This proves the lemma.
\end{proof}

\paragraph{Tomography result.}
The proof of Theorem~\ref{thm:tomo_rank} immediately follows from application of Theorem~\ref{thm:tomo_extent} and using Theorem~\ref{lem:ub_extent_rank} to comment on the upper bound on $\xi(\ket{\psi})$ for states $\ket{\psi}$ with stabilizer-rank $k$.

\section{Discussion}\label{sec:discussion}
In this note, we showed that states with stabilizer rank $O(\log(n))$ can be efficiently learned. However, the question remains if this can be extended to obtain efficient tomography algorithms for states with stabilizer rank $\poly(n)$. One way of achieving this result would be to obviously improve the upper bound on stabilizer extent for states $\ket{\psi}$ with stabilizer rank $k$ to $\xi(\ket{\psi}) \leq O(\poly(k))$. 
Of particular interest is also trying to obtain better lower bounds on the stabilizer rank of the magic state $\ket{T}^{\otimes n}$ where $\ket{T} = (\ket{0} + e^{i \pi/4} \ket{1})/\sqrt{2}$. Currently, the best known result is $\Omega(n^2/\polylog(n))$ due to~\cite{mehraban2024quadratic,lovitz2022new}. 

\begin{corollary}\label{cor:lower_bound_rank_Tstates}
Let $\alpha > \cos^2(\pi/8)$ be a constant. Every $n$-qubit state $\ket{\psi}$ satisfying $|\la \psi | T^{\otimes n} \ra|^2 \geq \alpha^n$ has stabilizer rank $\chi(\ket{\psi}) = \Omega(n)$.
\end{corollary}

The known extent formula and its multiplicativity for tensor products~\cite{bravyi2019simulation} give
\begin{equation}\label{eq:T-extent}
 \xi(\ket{T^{\otimes n}})= (\cos (\pi/8))^{-n}.
\end{equation}
The proof of Corollary~\ref{cor:lower_bound_rank_Tstates} then follows from Lemma~\ref{lem:ub_extent} and Eq.~\eqref{eq:T-extent}. This only however yield a linear rank lower bound. While it improves known lower bounds in this regime (see~\cite{kalra2025stabilizer} for more details), it does not improve the nearly-quadratic lower bounds for exact or constant-error approximations. What extent-rank relationship would then suffice?
From Eq.~\eqref{eq:T-extent}, we have that proving $\log\xi(\ket{\psi})\leq f(\chi(\ket{\psi}))$ would imply a lower bound of $f(\chi(\ket{\psi})) \geq \Omega(n)$. If for some fixed $\alpha>0$ one could prove
$\xi(\ket{\psi}) \leq \exp\!\bigl(O(\chi(\ket{\psi})^\alpha)\bigr)$ then
the resulting lower bound would be $\chi(\ket{T^{\otimes n}}) = \Omega(n^{1/\alpha})$. Taking $\alpha=1/2$ would give $\Omega(n^2)$, removing the logarithmic loss in the best known bounds. Any fixed $\alpha<1/2$ would give a super-quadratic lower bound. To obtain an exponential lower bound of $\xi(\ket{T^{\otimes n}}) = 2^{\Omega(n)}$, one route would be to then show that $\xi(\ket{\psi}) \leq \poly(\chi(\ket{\psi}))$ for every $n$-qubit state $\ket{\psi}$. Coincidentally, this is precisely the desired result for obtaining polynomial-time tomography protocols for states with stabilizer rank $\poly(n)$ mentioned earlier.

\paragraph{Acknowledgements.} 
AD thanks Srinivasan Arunachalam for multiple conversations on this topic and going through an initial draft of this note.

\paragraph{AI Disclosure.} 
The first proof of Theorem~\ref{lem:ub_extent_rank} was obtained using GPT 6 Astra. It was then realized by the author that it was closely related to the proof of convergence in energy increment arguments such as that in \cite{ad2026tomoextent}. The proof was subsequently rewritten and simplified by the author to reflect this. The author takes full responsibility for the correctness, exposition, and attribution of this work.


\bibliographystyle{alpha}
\bibliography{references}

\end{document}